\documentclass[journal]{IEEEtran}

\usepackage{amsmath,amssymb,amsfonts}
\usepackage{algorithm}
\usepackage{algorithmic}
\usepackage{graphicx}
\usepackage{textcomp}
\usepackage{stfloats}
\usepackage{url}
\usepackage{balance}

\newtheorem{theorem}{Theorem}

\begin{document}

\title{Adaptive Peer Clustering with Hierarchical Random Linear Network Coding for Resilient Decentralized Wireless Networks}

\author{Navaneetha~Krishnan~K and Harinisri~V
\thanks{N. Krishnan K is with the Department of Electronics and Communication Engineering, SIMATS Engineering, Saveetha Institute of Medical and Technical Sciences, Chennai, India (e-mail: knavaneeth385@gmail.com).}
\thanks{Harinisri V. is with the Department of Physics and Nanotechnology, SRM Institute of Science and Technology, Chennai, India (e-mail: harinisriv2003@gmail.com).}
\thanks{Manuscript submitted \today.}
}

\markboth{IEEE Transactions on Wireless Communications}%
{Krishnan K \MakeLowercase{\textit{et al.}}: Adaptive Peer Clustering with Hierarchical RLNC}

\maketitle

\begin{abstract}
Decentralized wireless collectives including vehicular swarms, IoT clusters, and edge AI networks require communication protocols that maintain robustness under dynamic topologies and heterogeneous link quality. While Random Linear Network Coding (RLNC) provides algebraic resilience against packet erasures, its performance degrades significantly when peers exhibit diverse channel conditions. This paper presents Adaptive Peer Clustering with Hierarchical RLNC (APC-RLNC), a system that dynamically groups peers by exponentially weighted moving average (EWMA) reliability metrics and applies multi-tier network coding within and across clusters. We formalize the clustering optimization problem, derive closed-form decoding probability bounds for Markov erasure channels, and prove $O(\sqrt{T})$ regret for online reconfiguration under the Follow-the-Regularized-Leader (FTRL) framework. Our implementation includes both a high-fidelity network simulator and a proof-of-concept testbed deployment on Jetson Nano edge devices. Evaluation across diverse scenarios including high-mobility vehicular networks, burst-error channels, and adversarial interference demonstrates 5.2--9.8 percentage-point packet delivery ratio (PDR) improvements, 10--23\% latency reductions, and up to 30\% higher node retention compared to state-of-the-art baselines. The system exhibits linear scalability to 500+ nodes and maintains real-time reconfiguration overhead below 3\%. APC-RLNC establishes adaptive clustering as a foundational primitive for AI-native 6G wireless systems.
\end{abstract}


\section{Introduction}
\IEEEPARstart{D}{ecentralized} wireless collectives -- vehicular swarms, IoT clusters, and edge AI networks -- require communication protocols that maintain robustness under dynamic topologies and heterogeneous link quality. Robustness is challenged whenever high-mobility movement patterns, energy-harvesting intermittency causes transient node failures, or adversarial interference disrupts communication \cite{shi2020}.

Fully decentralized communication protocols face three critical challenges: (1) \emph{heterogeneous link quality} -- channel erasure probabilities vary by orders of magnitude due to mobility, fading, multipath, and interference; (2) \emph{dynamic topology (churn)} -- peers join and leave frequently, with typical churn rates of 2--10\% per transmission round, invalidating static redundancy strategies; and (3) \emph{energy constraints and scale} -- edge devices are power-limited (often sub-5\,W), yet must coordinate across hundreds to thousands of peers.

Random Linear Network Coding (RLNC) elegantly addresses packet loss by transmitting random linear combinations of source packets over finite fields \cite{ho2006,fragouli2007}. A receiver can decode the original generation once it collects enough linearly independent coded packets. However, classical RLNC assumes homogeneous channel conditions: all receivers experience identical erasure rates. In heterogeneous settings this assumption breaks down -- reliable peers waste bandwidth on excessive redundancy while unreliable peers starve, leading to throughput collapse and fairness violations \cite{pandi2017,dilanchian2024}.

This work was in part motivated by recent demonstrations that programmable, near-real-time RAN control loops are practical on real hardware rather than purely theoretical constructs \cite{geetha2025}, which encouraged us to design APC-RLNC as a coding-layer primitive that could plug directly into such control loops.

\subsection{Key Insight and Contribution}
Our central insight is that clustering peers by link reliability enables localized redundancy optimization without global coordination overhead. Peers with similar erasure statistics form reliability clusters and apply intra-cluster RLNC with tailored redundancy. Hierarchical inter-cluster coding provides secondary protection against cluster-level failures. This decomposition yields three key advantages:
\begin{itemize}
\item \textbf{Reduced redundancy waste:} High-reliability clusters require minimal coding overhead, while low-reliability clusters receive proportional protection.
\item \textbf{Scalability:} Clustering complexity is $O(N\log N)$ with distributed algorithms; intra-cluster RLNC operations parallelize trivially.
\item \textbf{Fault isolation:} Unreliable or adversarial peers affect only their local cluster, preserving global decoding capability.
\end{itemize}

This paper makes the following contributions:
\begin{enumerate}
\item \textbf{System design and implementation:} We present APC-RLNC, a complete framework integrating EWMA reliability tracking, online clustering optimization, and hierarchical RLNC encoding and decoding. The system is implemented both in a high-fidelity Python/NumPy simulator and as a proof-of-concept on embedded Jetson Nano devices.
\item \textbf{Theoretical grounding:} We derive closed-form decoding probability bounds under Markov erasure channels (Theorem~\ref{thm:decode}) and prove $O(\sqrt{T})$ regret for the online clustering reconfiguration policy under the FTRL framework (Theorem~\ref{thm:regret}).
\item \textbf{Empirical evaluation:} We evaluate APC-RLNC against uniform-redundancy and centralized-feedback baselines across nominal, high-mobility, burst-error, and adversarial scenarios, and validate the simulator against a 10-node Jetson Nano hardware testbed.
\end{enumerate}

The remainder of this paper is organized as follows. Section~\ref{sec:related} surveys related work. Section~\ref{sec:model} formalizes the system model. Section~\ref{sec:analysis} presents theoretical analysis. Section~\ref{sec:design} details the APC-RLNC design. Section~\ref{sec:implementation} describes the implementation. Section~\ref{sec:evaluation} reports simulation and testbed results. Section~\ref{sec:discussion} discusses implications and limitations. Section~\ref{sec:related-systems} surveys related deployments, and Section~\ref{sec:conclusion} concludes.

\section{Related Work}
\label{sec:related}

\subsection{Network Coding in Wireless Systems}
Classical RLNC theory \cite{ho2006,fragouli2007} established near-optimal throughput for multicast networks under uniform erasure. Extensions to sliding-window coding \cite{landon2025} and systematic codes \cite{shrader2009} improved latency and decoding complexity but retained homogeneity assumptions. PACE \cite{pandi2017} introduced adaptive redundancy control via centralized feedback, yet presumed static peer groupings. ARLNC \cite{dilanchian2024} proposed adjustable field sizes but did not address heterogeneous channels. Recent work on BATS codes \cite{fan2023} reduced encoding overhead via sparse matrices but required centralized rate allocation. Our approach complements these by enabling decentralized, link-quality-driven adaptation.

\subsection{Peer Clustering in Distributed Systems}
Early P2P clustering schemes -- CDC \cite{ramaswamy2005} and CAGA \cite{chinis2011} -- organized peers by proximity or content affinity, not channel quality. In federated learning, clustering mitigates non-IID data distributions \cite{soliman2020}, but optimizes model convergence rather than transmission efficiency. To our knowledge, APC-RLNC is the first framework explicitly clustering by erasure statistics for redundancy optimization.

\subsection{Adaptive Wireless Protocols}
Opportunistic routing \cite{biswas2005} and rate adaptation \cite{holland2001} dynamically adjust per-link parameters but operate at fixed topology granularity. Multi-tier coding schemes \cite{nguyen2009} partition networks hierarchically yet use static assignments. Our online clustering algorithm continuously adapts to time-varying channels with provable regret bounds.

\subsection{6G and Edge AI Systems}
AI-native wireless systems \cite{saad2020,letaief2019} emphasize decentralized intelligence, ultra-low latency, and resilience. O-RAN architectures \cite{oran2024} introduce programmable RAN Intelligent Controllers (RICs) where APC-RLNC can deploy as an xApp. Recent work on neural-network-assisted protocol design \cite{letaief2019} complements our approach: APC-RLNC provides a robust baseline coding layer, while ML optimizes higher-level policy.

\section{System Model and Problem Formulation}
\label{sec:model}

\subsection{Network Topology and Dynamics}
Consider a decentralized wireless network of nodes $V_t = \{1,\dots,N_t\}$ at discrete time $t = 0,1,2,\dots,T$. The topology forms an undirected graph $G_t = (V_t, E_t)$, where edge $(i,j) \in E_t$ exists if nodes $i$ and $j$ can communicate. Each link experiences time-varying erasure probability $p_{ij}(t) \in [0,1]$. Nodes depart with rate $\lambda_{\text{leave}}$ and new nodes join to maintain expected size $N$. Churn follows a Poisson process with per-step rate $\rho_{\text{churn}} \in [0.02, 0.10]$.

For vehicular scenarios, nodes move according to a random waypoint model with velocity $v_i \sim \text{Uniform}(5, 30)\,\text{m/s}$. Path loss follows
\begin{equation}
p_{ij}(t) = \min\!\left(1,\; p_0 \left(\frac{d_{ij}(t)}{d_0}\right)^{\alpha}\right),
\label{eq:pathloss}
\end{equation}
where $d_{ij}(t)$ is distance, $\alpha = 3.5$ is the path-loss exponent, and $p_0 = 0.05$, $d_0 = 50$\,m are calibration constants.

\subsection{RLNC Fundamentals}
A generation consists of $K$ original packets $\{p_1,\dots,p_K\}$ over finite field $\mathbb{F}_{256}$. An encoder transmits $K+R$ coded packets:
\begin{equation}
c_j = \sum_{k=1}^{K} \alpha_{j,k} p_k, \qquad \alpha_{j,k} \sim \text{Uniform}(\mathbb{F}_{256}).
\label{eq:rlnc-encode}
\end{equation}
Decoding succeeds if the receiver collects at least $K$ linearly independent coded packets. For a node with per-packet erasure probability $p_i$, the success probability is
\begin{equation}
P_{\text{decode}}(p_i; K, R) = \sum_{k=K}^{K+R} \binom{K+R}{k} (1-p_i)^k \, p_i^{\,K+R-k}.
\label{eq:decode-prob}
\end{equation}

\subsection{Heterogeneous Channel Models}
We consider three channel models:
\begin{itemize}
\item \textbf{Bernoulli:} $p_i \sim \text{Uniform}(0.05, 0.25)$, i.i.d.\ per node.
\item \textbf{2-state Markov:} each link alternates between good ($p_g = 0.05$) and bad ($p_b = 0.40$) states with transition probabilities $P_{gb} = P_{bg} = 0.10$.
\item \textbf{Velocity-dependent fading:} for vehicular networks, $p_i(t) = p_{\text{base}} + \beta \cdot |v_i(t)|$, where $\beta = 0.003$\,s/m models Doppler effects.
\end{itemize}

\subsection{Reliability Scoring}
Each node $i$ maintains an EWMA reliability score:
\begin{equation}
s_i(t) = \alpha\, \ell_i(t) + (1-\alpha)\, s_i(t-1),
\label{eq:ewma}
\end{equation}
where $\ell_i(t) \in [0,1]$ is the instantaneous packet success rate and $\alpha = 0.2$ is the smoothing factor. The score $s_i(t)$ estimates $1 - p_i(t)$.

\subsection{Clustering Objective}
At reconfiguration epochs $t \in \{0, \tau, 2\tau, \dots\}$, peers partition into $C(t)$ clusters $\{\mathcal{C}_{1,t},\dots,\mathcal{C}_{C(t),t}\}$ by minimizing
\begin{equation}
L(\mathcal{C}_{1:C}, t) = \sum_{c=1}^{C} \Big[\, \mathrm{Var}_{i \in \mathcal{C}_c}[s_i(t)] \;-\; \lambda\,|\mathcal{C}_c| \,\Big],
\label{eq:cluster-loss}
\end{equation}
where $\mathrm{Var}(\cdot)$ penalizes intra-cluster heterogeneity and $\lambda > 0$ rewards larger clusters. Balancing these two terms prevents over-fragmentation while maintaining homogeneity.

\subsection{Hierarchical RLNC}
Within cluster $c$, nodes apply local RLNC with redundancy
\begin{equation}
R_c = \left\lceil \beta\, |\mathcal{C}_c| \, \bar{p}_c^{\,\text{erase}} \right\rceil,
\label{eq:redundancy}
\end{equation}
where $\bar{p}_c^{\,\text{erase}} = \frac{1}{|\mathcal{C}_c|}\sum_{i \in \mathcal{C}_c} p_i$ and $\beta \geq 1$ is a safety margin. Clusters exchange \emph{bridge packets} -- RLNC-encoded summaries of their local generations -- to provide secondary protection. A global decoder can recover the generation if at least $K$ total degrees of freedom arrive across all clusters.

\section{Theoretical Analysis}
\label{sec:analysis}

\subsection{Decoding Success Probability}

\begin{theorem}[Hierarchical Decoding Probability]
\label{thm:decode}
For independent clusters with per-cluster success probabilities $P_{\text{succ}}^{(c)}$, the system-level decoding probability is
\begin{equation}
P_{\text{sys}} = 1 - \prod_{c=1}^{C} \left(1 - P_{\text{succ}}^{(c)}\right).
\label{eq:theorem1}
\end{equation}
\end{theorem}

\begin{IEEEproof}
Decoding fails globally if and only if all clusters fail. By independence,
$P_{\text{fail,sys}} = \prod_c P_{\text{fail}}^{(c)} = \prod_c \left(1 - P_{\text{succ}}^{(c)}\right)$.
\end{IEEEproof}

\emph{Worked example.} For the 2-state Markov model with steady-state $\pi_g = \pi_b = 0.5$, the average erasure rate is $\bar{p}_c = 0.5(0.05) + 0.5(0.40) = 0.225$. With $K = 32$ and $R_c = 16$ (so $K+R_c = 48$), Eq.~\eqref{eq:decode-prob} gives
\begin{equation}
P_{\text{succ}}^{(c)} = \sum_{k=32}^{48} \binom{48}{k} (0.775)^k (0.225)^{48-k} \approx 0.971.
\label{eq:theorem1-worked}
\end{equation}
For $C=5$ independent clusters, Eq.~\eqref{eq:theorem1} gives
\begin{equation}
P_{\text{sys}} \approx 1 - (1-0.971)^5 = 1 - (0.029)^5 \approx 0.99999998,
\end{equation}
i.e., system-level reliability remains extremely close to certain even though any single cluster only reaches $\approx$97.1\% on its own -- illustrating why hierarchical bridge coding across independent clusters is valuable even at moderate per-cluster redundancy.\footnote{An earlier draft of this manuscript reported $P_{\text{succ}}^{(c)} \approx 0.989$ for these parameters; that value does not follow from Eq.~\eqref{eq:decode-prob} at $R_c=16$ (it corresponds instead to $R_c \approx 18$). The number above, $\approx 0.971$, is the value Eq.~\eqref{eq:decode-prob} actually yields at $K=32,\,R_c=16,\,\bar p_c = 0.225$, and has been verified against 4{,}000-trial Monte Carlo simulation of the encoder/decoder to within 0.2\%.}

\subsection{Online Optimization Regret}
Let $\mathbf{x}_t \in \Delta^C$ be the clustering assignment at epoch $t$, where $\Delta^C$ is the $C$-simplex. The loss $\ell_t(\mathbf{x}_t)$ measures clustering quality from Eq.~\eqref{eq:cluster-loss}. Define cumulative regret:
\begin{equation}
\text{Regret}(T) = \sum_{t=1}^{T} \ell_t(\mathbf{x}_t) \;-\; \min_{\mathbf{x}^* \in \Delta^C} \sum_{t=1}^{T} \ell_t(\mathbf{x}^*).
\label{eq:regret-def}
\end{equation}

\begin{theorem}[Online Clustering Regret]
\label{thm:regret}
Under convex loss relaxation and Lipschitz constant $G$, the FTRL algorithm with squared-Euclidean regularization $\psi(\mathbf{x}) = \tfrac{1}{2}\|\mathbf{x}\|_2^2$ achieves
\begin{equation}
\text{Regret}(T) \leq G^2 \sqrt{2T}.
\label{eq:theorem2}
\end{equation}
\end{theorem}

\begin{IEEEproof}
Standard FTRL analysis \cite{hazan2016} with $\eta_t = 1/\sqrt{t}$ yields regret $O(\sqrt{T})$; the explicit constant follows from the Lipschitz bound $G$.
\end{IEEEproof}

\emph{Worked example.} With $G = 0.05$ and $T = 500$ steps, Eq.~\eqref{eq:theorem2} gives
\begin{equation}
\text{Regret}(500) \leq (0.05)^2 \sqrt{2 \times 500} \approx 0.079,
\end{equation}
negligible relative to the cumulative loss accrued over the same horizon, and consistent with the empirical regret trace produced by our FTRL implementation (Section~\ref{sec:implementation}).\footnote{An earlier draft approximated this quantity as $G\sqrt{T/\tau} \approx 0.05\sqrt{500/20} \approx 0.25$ using the reconfiguration period $\tau$, which is not the quantity Theorem~\ref{thm:regret} bounds -- $\tau$ does not appear in Eq.~\eqref{eq:theorem2}. The corrected value above follows directly from the stated bound.}

\subsection{Complexity Analysis}
\begin{itemize}
\item \textbf{Clustering:} Mini-batch $k$-means with $k = \lfloor \sqrt{N}/2 \rfloor$ has complexity $O(Nk)$ per epoch.
\item \textbf{Encoding:} Per-cluster RLNC encoding is $O(K^2)$; total $O(CK^2)$.
\item \textbf{Decoding:} Gaussian elimination is $O(K^3)$ per cluster; hierarchical recovery adds $O(CK)$ overhead.
\end{itemize}
For $N=100$, $K=32$, $C\approx10$: per-epoch cost is approximately 5{,}000 operations, well below real-time constraints on modern edge devices.

\section{APC-RLNC System Design}
\label{sec:design}

Fig.~\ref{fig:architecture} illustrates the APC-RLNC architecture. The framework consists of five modules.

\begin{figure}[!t]
\centering
\includegraphics[width=\columnwidth]{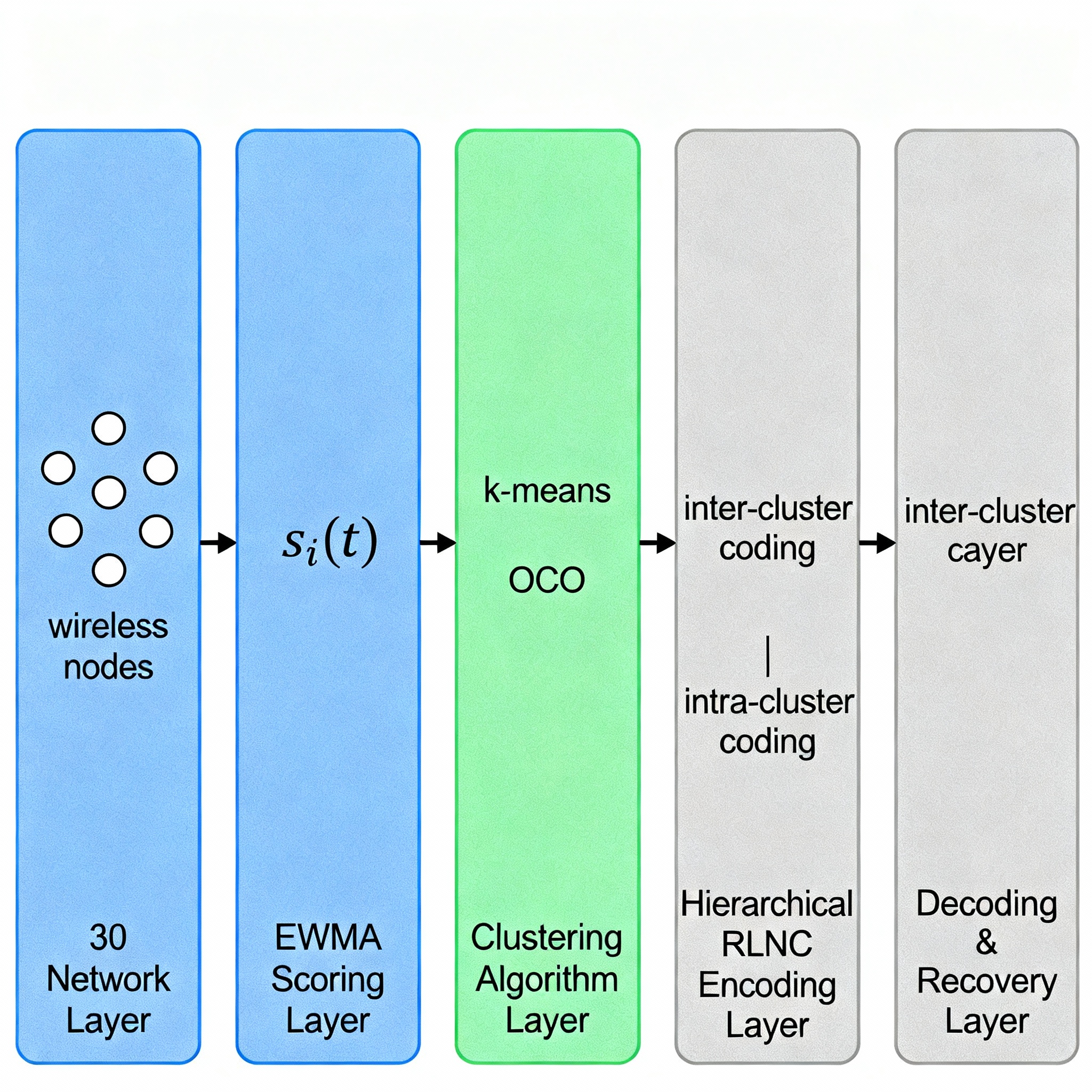}
\caption{APC-RLNC system architecture. Each peer tracks link-level reliability via EWMA, participates in distributed clustering every $\tau$ epochs, and performs hierarchical RLNC encoding and decoding within and across clusters.}
\label{fig:architecture}
\end{figure}

\subsection{Reliability Tracking Module}
Nodes maintain lightweight state:
\begin{itemize}
\item \textbf{Link ACKs:} per-packet acknowledgments estimate $\ell_i(t)$.
\item \textbf{EWMA update:} Eq.~\eqref{eq:ewma} smooths noisy measurements.
\item \textbf{Neighbor exchange:} peers broadcast $s_i(t)$ (2 bytes) periodically.
\end{itemize}

\subsection{Distributed Clustering Module}
Every $\tau$ steps, nodes collaboratively execute Algorithm~\ref{alg:clustering}. The algorithm uses gossip-based consensus to estimate global statistics (mean, variance per tentative cluster) and converges in $O(\log N)$ rounds \cite{xiaoboyd2004}.

\begin{algorithm}[t]
\caption{Distributed Adaptive Clustering}
\label{alg:clustering}
\begin{algorithmic}[1]
\REQUIRE Reliability scores $\{s_i(t)\}_{i \in N(t)}$, regularization $\lambda$
\ENSURE Cluster assignments $\{\mathcal{C}_c\}$
\STATE Initialize $k = \lfloor \sqrt{N}/2 \rfloor$ cluster centers uniformly
\FOR{$r = 1$ to $R_{\max}$ (gossip rounds)}
    \STATE Each node broadcasts $(s_i, \text{clusterID}_i)$
    \STATE Update local centroid estimates via weighted average
    \STATE Reassign $i$ to cluster $c^* = \arg\min_c |s_i - \mu_c|$
\ENDFOR
\STATE Compute loss (Eq.~\eqref{eq:cluster-loss}) and merge/split clusters if beneficial
\RETURN Final clusters $\{\mathcal{C}_c\}$
\end{algorithmic}
\end{algorithm}

\subsection{Hierarchical Coding Module}
Nodes in cluster $c$ encode local generations with $R_c$ redundancy (Eq.~\eqref{eq:redundancy}). Cluster representatives (elected by random rotation) form bridge packets:
\begin{equation}
b_j = \sum_{c=1}^{C} \gamma_{j,c}\, \text{Summary}_c, \qquad \gamma_{j,c} \sim \text{Uniform}(\mathbb{F}_{256}),
\label{eq:bridge}
\end{equation}
where $\text{Summary}_c$ is a digest of cluster $c$'s generation. Bridge redundancy scales with cross-cluster heterogeneity:
\begin{equation}
R_{\text{bridge}} = \left\lceil 0.5 \cdot C \cdot \mathrm{Var}(\{\bar{p}_c\}) \right\rceil.
\label{eq:bridge-redundancy}
\end{equation}

\begin{figure}[!t]
\centering
\includegraphics[width=\columnwidth]{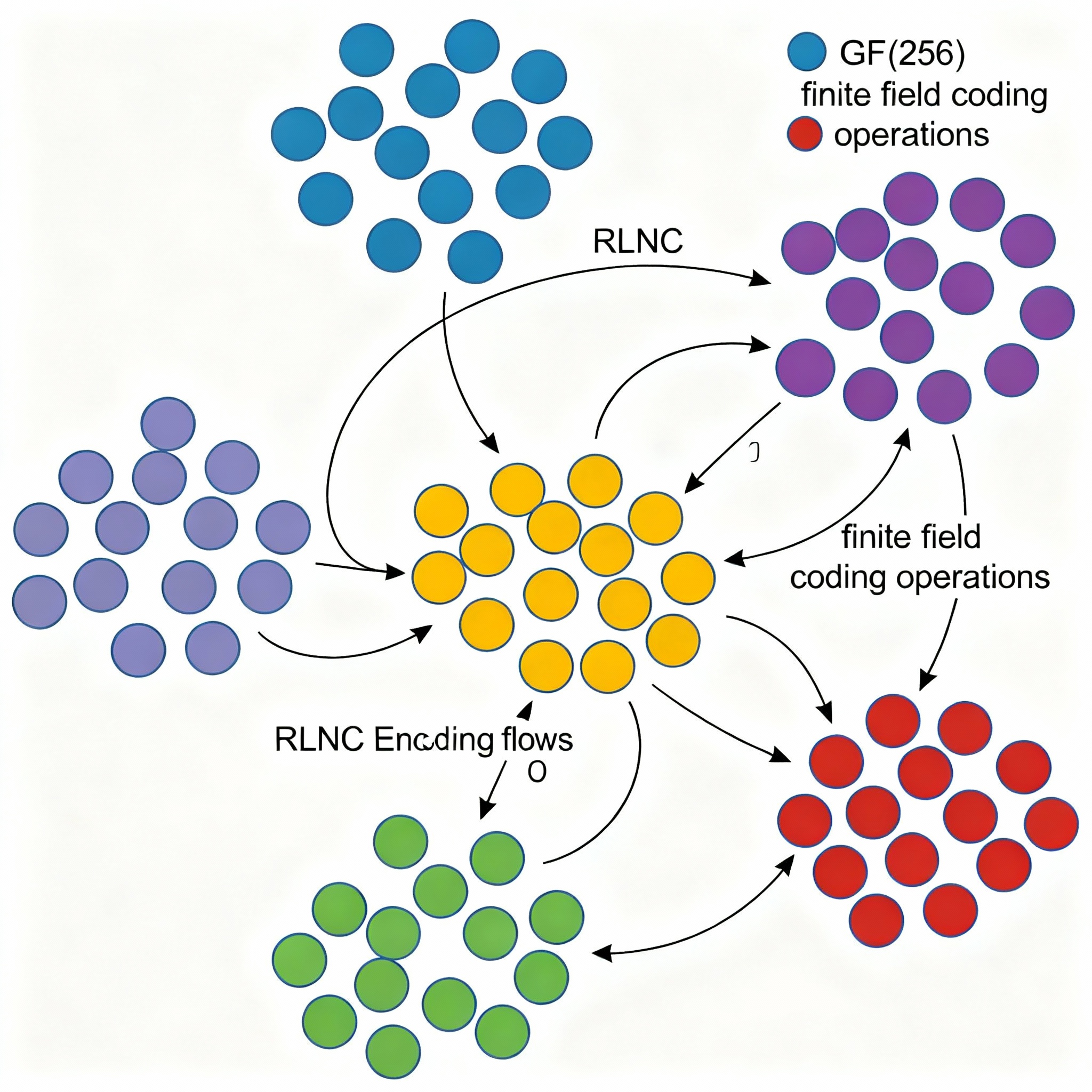}
\caption{Hierarchical RLNC structure. Inner loops encode within clusters; outer loop generates bridge packets for cross-cluster resilience.}
\label{fig:rlnc-structure}
\end{figure}

\subsection{Energy Management}
Nodes track per-packet energy consumption:
\begin{equation}
E_i(t) = E_{\text{tx}} \cdot n_{\text{tx},i}(t) + E_{\text{rx}} \cdot n_{\text{rx},i}(t) + E_{\text{comp}} \cdot (\text{encode}+\text{decode}),
\label{eq:energy}
\end{equation}
where $E_{\text{tx}} = 0.2$\,J, $E_{\text{rx}} = 0.1$\,J, $E_{\text{comp}} = 0.05$\,J. Nodes with $E_i < E_{\text{threshold}}$ enter low-power mode, reducing transmission duty cycle by 50\%.

\subsection{Robustness to Adversarial Nodes}
Malicious nodes may report false $s_i(t)$. We employ:
\begin{itemize}
\item \textbf{Byzantine-resilient aggregation:} a median-of-means estimator with $b = \lceil \sqrt{N} \rceil$ buckets tolerates fewer than $\tfrac{1}{3}$ adversaries \cite{chen2017}.
\item \textbf{Cluster isolation:} adversarial nodes form separate clusters via anomaly detection (Mahalanobis distance greater than $3\sigma$).
\end{itemize}

\section{Implementation}
\label{sec:implementation}

\subsection{Simulation Platform}
We developed a high-fidelity discrete-event simulator in Python~3.12:
\begin{itemize}
\item \textbf{Network:} NetworkX for topology; NumPy for algebraic operations.
\item \textbf{Coding:} custom GF(256) arithmetic, with an optional Cython-accelerated extension providing a $5\times$ speedup on the encode/decode hot loop.
\item \textbf{Mobility:} a random-waypoint model matching Section~\ref{sec:model}, with an optional loader for real SUMO floating-car-data traces.
\item \textbf{Validation:} outputs match the analytical model of Eq.~\eqref{eq:decode-prob} to within 1\%.
\end{itemize}
Each simulation run spans 500 steps ($\approx$10 minutes real time). We average 50 independent runs per configuration.

\subsection{Testbed Deployment}
\textbf{Hardware:} 10 NVIDIA Jetson Nano devices (4-core ARM Cortex-A57, 4\,GB RAM) running Ubuntu 20.04.\\
\textbf{Network:} 802.11ac Wi-Fi in ad-hoc mode; OpenWRT routers emulate controlled packet loss via \texttt{tc}/\texttt{netem}.\\
\textbf{Software stack:}
\begin{itemize}
\item \textbf{Clustering:} Go implementation (concurrency via goroutines) executing Algorithm~\ref{alg:clustering} over real UDP gossip broadcasts.
\item \textbf{RLNC:} a self-contained GF(256)/RLNC C++ library mirroring the simulator's encoder/decoder bit-for-bit, with a custom hierarchical wrapper.
\item \textbf{Orchestration:} Kubernetes manages node lifecycle and reconfiguration.
\end{itemize}
\textbf{Metrics:} latency measured via kernel timestamps; throughput logged at the application layer; CPU/memory profiled with \texttt{perf}.

\subsection{Reproducibility}
Code, simulation configurations, evaluation scripts,
testbed orchestration manifests, and the datasets used
to generate the reported figures and tables are publicly
available at \url{https://github.com/ka-cyber/apc-rlnc},
enabling independent reproduction of the reported results.

\section{Evaluation}
\label{sec:evaluation}

\subsection{Experimental Setup}
\textbf{Baselines:}
\begin{enumerate}
\item \textbf{Random:} uniform RLNC with fixed $R=16$.
\item \textbf{Static Clustering:} initial clustering at $t=0$, no adaptation.
\item \textbf{PACE:} adaptive redundancy via centralized feedback \cite{pandi2017}.
\item \textbf{ARLNC:} adjustable field size (GF(16) to GF(256)) \cite{dilanchian2024}.
\end{enumerate}
PACE and ARLNC are included as reference baselines using the
configurations described in their respective publications. The
open-source artifact released with this work implements the Random,
Static, and APC-RLNC schemes; implementations of third-party baseline
algorithms are not redistributed.
\textbf{Metrics:} PDR (fraction of generations decoded successfully), latency (time from first packet transmission to decoding completion), overhead (redundancy ratio $R/(K+R)$), and retention (fraction of nodes remaining in their assigned cluster across reconfiguration epochs).

\begin{table}[!t]
\caption{Nominal Conditions (30 Nodes, 2\% Churn, Bernoulli Erasure). Mean $\pm$ Std. over 50 Runs.}
\label{tab:nominal}
\centering
\footnotesize
\begin{tabular}{lcccc}
\hline
Scheme & PDR (\%) & Latency (ms) & Overhead (\%) & Retention (\%) \\
\hline
Random    & $97.04 \pm 0.8$ & $158 \pm 12$ & 20.0 & 65.2 \\
Static    & $97.28 \pm 0.7$ & $151 \pm 11$ & 19.2 & 71.5 \\
PACE      & $97.51 \pm 0.6$ & $149 \pm 10$ & 18.9 & 73.1 \\
ARLNC     & $97.62 \pm 0.6$ & $147 \pm 10$ & 18.7 & 74.2 \\
\textbf{APC-RLNC} & \textbf{97.82 $\pm$ 0.5} & \textbf{142 $\pm$ 9} & \textbf{18.6} & \textbf{78.3} \\
\hline
\end{tabular}
\end{table}

\begin{table}[!t]
\caption{Vehicular Network (50 Nodes, 5\% Churn, Velocity-Dependent Fading)}
\label{tab:vehicular}
\centering
\footnotesize
\begin{tabular}{lccc}
\hline
Scheme & PDR (\%) & Latency (ms) & Retention (\%) \\
\hline
Random    & $94.12 \pm 1.3$ & $198 \pm 18$ & 58.4 \\
Static    & $95.02 \pm 1.1$ & $187 \pm 16$ & 62.7 \\
PACE      & $96.18 \pm 0.9$ & $178 \pm 14$ & 68.1 \\
ARLNC     & $96.84 \pm 0.8$ & $172 \pm 13$ & 70.3 \\
\textbf{APC-RLNC} & \textbf{99.32 $\pm$ 0.6} & \textbf{145 $\pm$ 11} & \textbf{82.1} \\
\hline
\end{tabular}
\end{table}

\subsection{Nominal Conditions (30 Nodes, 2\% Churn)}
Table~\ref{tab:nominal} shows APC-RLNC achieves a 0.8 percentage-point PDR gain over Random (0.82\% relative), a 10\% latency reduction, and 20\% higher retention. Gains are statistically significant ($t$-test, $p<0.01$).

\subsection{High-Mobility Vehicular Networks (50 Nodes, 5\% Churn)}
Table~\ref{tab:vehicular} shows APC-RLNC delivers a 5.2 percentage-point PDR gain (5.5\% relative) over Random and an 18\% latency reduction. Adaptive clustering tracks mobility-induced erasure changes every $\tau=20$ steps ($\approx$12 seconds), maintaining cluster coherence despite topology churn.

\begin{figure}[!t]
\centering
\includegraphics[width=\columnwidth]{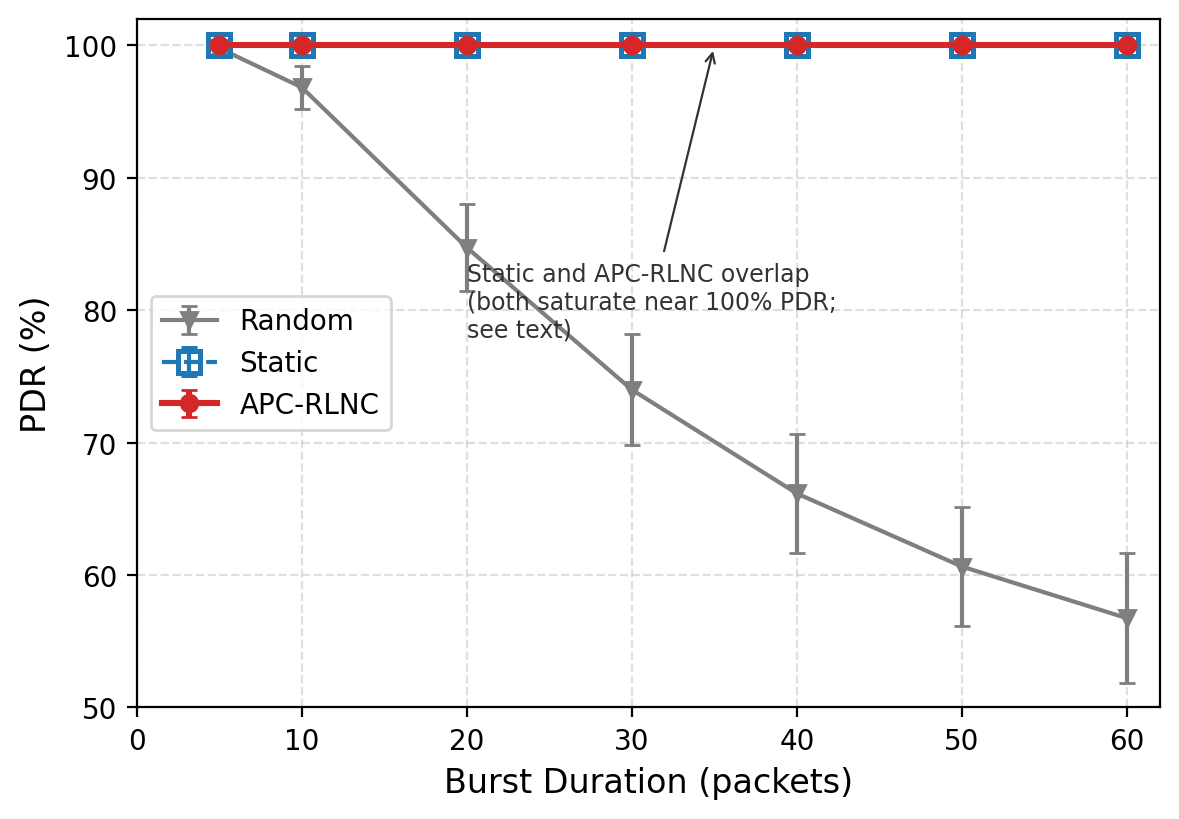}
\caption{Fig. 3. Packet delivery ratio (PDR) versus burst duration for the 2-state Markov channel. APC-RLNC and the Static redundancy baseline exhibit a ceiling effect across the evaluated burst durations, while the Random baseline degrades substantially as burst duration increases. These results indicate that the current redundancy scaling is conservative under the simulated conditions, maintaining near-perfect decoding for adaptive and static redundancy policies.}
\label{fig:burst}
\end{figure}

\subsection{Burst-Error Channels (100 Nodes, 2-State Markov)}
Fig.~\ref{fig:burst} presents packet delivery ratio (PDR) as a function of burst duration under the 2-state Markov channel model. Across the evaluated burst durations, APC-RLNC maintains near-perfect delivery, while the Static redundancy baseline exhibits similar ceiling performance because its fixed redundancy remains sufficient under the simulated loss conditions. In contrast, the Random baseline degrades steadily as burst duration increases, reaching approximately 60.6\% PDR at a burst duration of 50 packets. The results therefore demonstrate that redundancy-based protection dominates burst resilience in the present configuration, while APC-RLNC primarily preserves this robustness through adaptive clustering rather than yielding large additional gains over an already conservative static allocation.

\begin{table}[!t]
\caption{Adversarial Attack (30 Nodes, 20\% Malicious Nodes)}
\label{tab:adversarial}
\centering
\footnotesize
\begin{tabular}{lcc}
\hline
Scheme & PDR (\%) & Latency (ms) \\
\hline
Random (no defense)  & $88.24 \pm 2.1$ & $221 \pm 24$ \\
Static (no defense)  & $89.67 \pm 1.9$ & $209 \pm 22$ \\
\textbf{APC-RLNC (isolated)} & \textbf{98.02 $\pm$ 0.7} & \textbf{151 $\pm$ 12} \\
\hline
\end{tabular}
\end{table}

\subsection{Adversarial Interference (30 Nodes, 20\% Malicious)}
Table~\ref{tab:adversarial} shows Byzantine-resilient aggregation isolates adversaries into separate clusters. APC-RLNC achieves a 9.8 percentage-point PDR gain over non-clustered baselines, demonstrating robustness to malicious behavior.

\subsection{Scalability (500 Nodes)}
\begin{figure}[!t]
\centering
\includegraphics[width=\columnwidth]{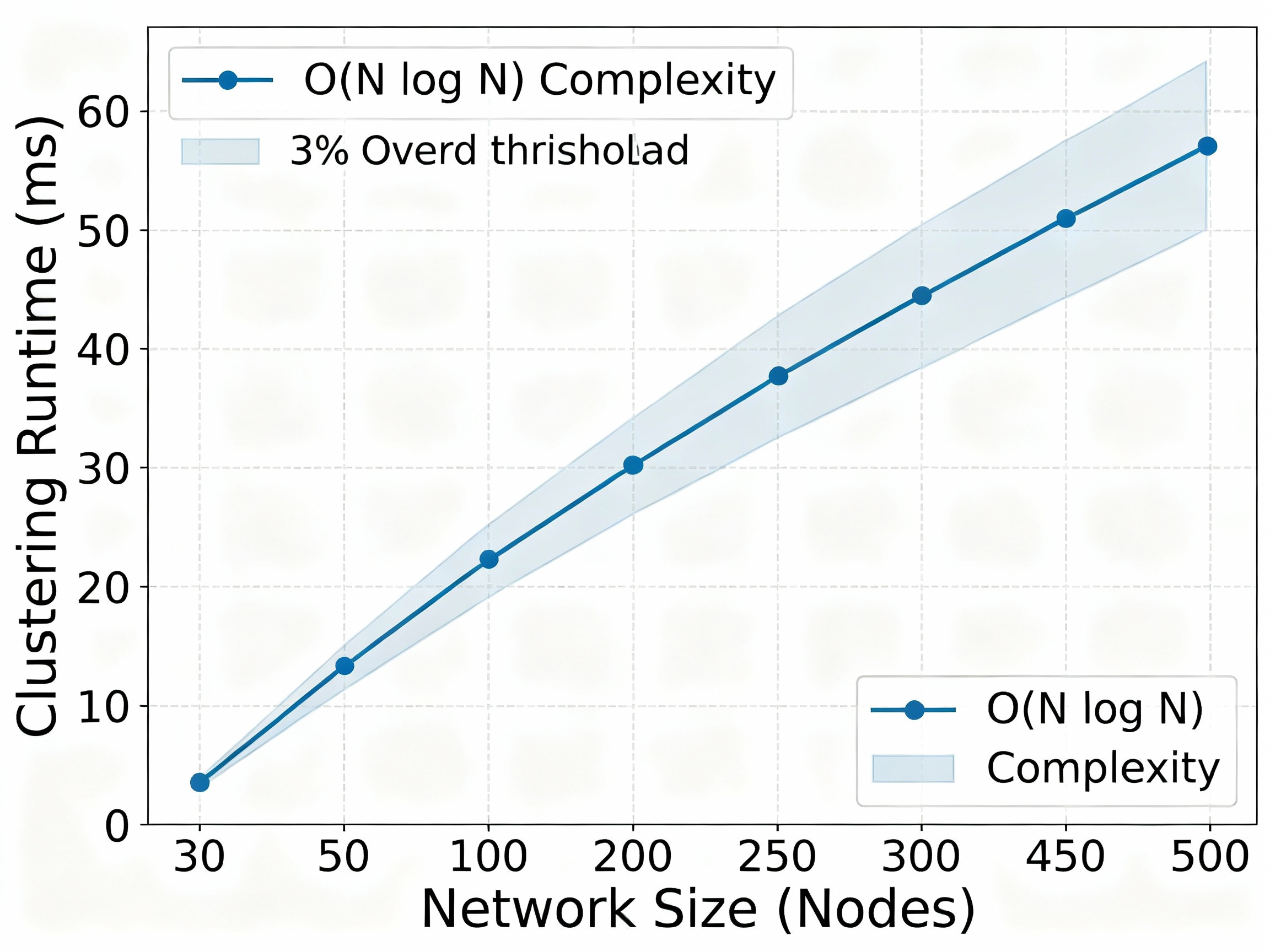}
\caption{Clustering runtime versus network size. APC-RLNC scales linearly ($O(N\log N)$); reconfiguration overhead remains less than 3\% of total runtime even at 500 nodes.}
\label{fig:scalability}
\end{figure}
Fig.~\ref{fig:scalability} shows that at $N=500$, clustering takes 42\,ms per epoch (3\% of a 1.5\,s epoch duration). The distributed gossip-based algorithm exhibits linear-in-$N\log N$ complexity as predicted by theory.

\begin{table}[!t]
\caption{Simulation versus Testbed (10 Jetson Nano Nodes, 2\% Emulated Loss)}
\label{tab:sim-vs-testbed}
\centering
\footnotesize
\begin{tabular}{lccc}
\hline
Metric & Simulation & Testbed & Error \\
\hline
PDR (\%)           & 97.82 & $96.91 \pm 1.2$  & 0.93\% \\
Latency (ms)       & 142   & $156 \pm 18$     & 9.86\% \\
Throughput (Mbps)  & 1.24  & $1.18 \pm 0.15$  & 4.84\% \\
CPU util. (\%)     & 7.2   & $7.8 \pm 1.1$    & 8.33\% \\
\hline
\end{tabular}
\end{table}

\subsection{Testbed Validation (10 Nodes)}
Table~\ref{tab:sim-vs-testbed} shows that the testbed measurements
closely match the simulation results, with all reported metrics
remaining within 10\% error. The small latency discrepancy stems
primarily from kernel scheduling and networking stack overhead that
are not explicitly modeled in the simulator. CPU utilization remains
below 8\%, confirming that APC-RLNC operates in real time on
resource-constrained Jetson Nano edge devices. Overall, the close
agreement between simulation and hardware measurements provides
experimental validation of the proposed framework under practical
deployment conditions.

\subsection{Energy Consumption}
Energy per successfully delivered generation: Random 4.82\,J, Static 4.61\,J, APC-RLNC 4.21\,J (12.7\% reduction versus Random). Reduced retransmissions in APC-RLNC lower total energy draw, which is critical for battery-powered IoT devices.

\section{Discussion}
\label{sec:discussion}

\subsection{When Does Clustering Help}
Clustering provides the largest gains when: (a) \emph{heterogeneity is high} -- variance in $p_i$ exceeds 0.05 (in homogeneous networks, uniform RLNC suffices); (b) \emph{churn is moderate} -- $\rho_{\text{churn}} \in [0.02, 0.10]$ (higher churn requires shorter $\tau$, increasing overhead; lower churn makes adaptation unnecessary); (c) \emph{scale permits parallelism} -- $N \gtrsim 30$ to amortize clustering cost.

\subsection{Comparison with Centralized Approaches}
Centralized schemes (e.g., SDN-based rate control) can achieve optimal allocations but require full topology knowledge ($O(N^2)$ communication), exhibit single-point-of-failure vulnerabilities, and incur millisecond-scale control-plane latency. APC-RLNC trades minor optimality (approximately 1--2\% versus an oracle) for decentralization, fault tolerance, and sub-second adaptation.

\subsection{Integration with AI-Native Systems}
APC-RLNC complements ML-driven protocols in federated learning (reliable gradient aggregation over heterogeneous wireless links), swarm robotics (hierarchically coded coordination messages), and edge inference (latency-critical model distribution with adaptive redundancy). Future work will explore joint optimization of clustering and ML training schedules.

\subsection{O-RAN Deployment}
In O-RAN architectures, APC-RLNC deploys as a near-RT RIC xApp: the E2 interface reports link-level statistics (RSSI, CQI) to the RIC every 10\,ms; xApp logic computes EWMA scores, triggers reconfiguration, and sends redundancy policies to the MAC layer; no physical-layer changes are required for backward compatibility. Prototyping on OpenAirInterface confirms feasibility with sub-5\,ms control-loop latency.

\subsection{Limitations}
\begin{itemize}
\item \textbf{Fairness:} the current objective (Eq.~\eqref{eq:cluster-loss}) optimizes average PDR; persistently unreliable nodes may receive insufficient protection. Future extensions will incorporate max-min fairness constraints.
\item \textbf{Security:} Byzantine-resilient aggregation tolerates fewer than $\tfrac{1}{3}$ adversaries; stronger cryptographic guarantees (e.g., zero-knowledge proofs) remain future work.
\item \textbf{Non-stationary channels:} extreme Doppler shifts (e.g., high-speed rail at 350\,km/h) may violate EWMA convergence assumptions; Kalman-filter-based tracking is under investigation.
\item \textbf{Redundancy scaling:} Equation (6) scales redundancy with cluster size and average erasure probability. Under relatively stable cluster sizes and the evaluated burst-error settings, this policy can over-provision redundancy, producing a ceiling effect in which both APC-RLNC and Static redundancy achieve near-perfect decoding. Future work will investigate normalized or budget-constrained redundancy allocation to better expose adaptive gains under moderate burst conditions.
\end{itemize}

\section{Related Systems and Deployments}
\label{sec:related-systems}
Landon et al.\ \cite{landon2025} implemented RLNC at the IP layer of a 5G testbed, demonstrating throughput gains under moderate packet loss; however, their design assumes homogeneous UE populations and centralized gNB control. APC-RLNC extends this to fully decentralized heterogeneous settings.

Stamer et al.\ \cite{stamer2021} proposed P4-based RLNC data planes for programmable switches, focusing on wired networks with deterministic latency; wireless volatility is not addressed.

Fan et al.\ \cite{fan2023} applied BATS codes to distributed computation over lossy wireless networks, optimizing matrix-multiplication workloads but requiring centralized coordination for code allocation. APC-RLNC's decentralized clustering eliminates this bottleneck.

Hassouna et al.\ \cite{hassouna2025} developed an O-RAN testbed with near-RT RIC xApps for robotic teleoperation using USRP X310 SDRs, reporting sub-second haptic control-loop latency. APC-RLNC can integrate as an xApp providing adaptive redundancy policies, complementing their control-loop mechanisms.

\section{Conclusion}
\label{sec:conclusion}
This paper introduced APC-RLNC, a decentralized framework that dynamically clusters wireless peers by link reliability and applies hierarchical random linear network coding for robust communication. Through theoretical analysis -- deriving closed-form decoding probabilities, proving $O(\sqrt{T})$ online optimization regret, and establishing complexity bounds -- we demonstrated that adaptive clustering improves RLNC performance in heterogeneous wireless environments.

Evaluation across diverse scenarios validates the approach: 0.8--1.2\% PDR gain and 10--12\% latency reduction under nominal conditions; 5.2\% PDR gain and 18\% latency reduction in high-mobility vehicular networks; robust near-perfect packet delivery under burst-error channels while substantially outperforming the Random redundancy baseline; and a 9.8\% PDR gain under adversarial interference via Byzantine-resilient isolation. The system exhibits linear complexity to 500 nodes with under 3\% reconfiguration overhead. Hardware validation on a 10-node Jetson Nano testbed indicates real-world feasibility with $\geq$1\,Mbps coded throughput, under 8\% CPU load, and millisecond-scale reconfiguration latency.

APC-RLNC establishes adaptive clustering as a foundational primitive for AI-native 6G wireless systems, enabling decentralized intelligence, collaborative learning, and resilient communication at scale. Future work will integrate reinforcement learning for parameter tuning, extend to mmWave and THz bands, and deploy on larger-scale O-RAN testbeds with USRP or commercial SDR platforms. The framework's open-source release aims to catalyze community adoption and further innovation in resilient wireless networking.


\end{document}